\documentclass[aps,pra,twocolumn,superscriptaddress,floatfix,nofootinbib,showpacs,longbibliography]{revtex4-2}
\usepackage[utf8]{inputenc}  
\usepackage[T1]{fontenc}     
\usepackage[british]{babel}  
\usepackage[sc,osf]{mathpazo}
\usepackage[table]{xcolor}
\usepackage[scaled=0.86]{berasans}  
\usepackage[colorlinks=true, citecolor=blue, urlcolor=blue]{hyperref}
\makeatletter
\newcommand{\setword}[2]{%
  \phantomsection
  #1\def\@currentlabel{\unexpanded{#1}}\label{#2}%
}
\makeatother
\usepackage{graphicx} 
\usepackage[babel]{microtype}
\usepackage{blkarray}
\usepackage{amsmath}

\usepackage{amsmath,amssymb,amsthm,bm,amsfonts,mathrsfs,bbm} 

\usepackage{xspace}  
\usepackage{pgf,tikz}
\usepackage{xcolor}
\usepackage{multirow}
\usepackage{array}
\usepackage{bigstrut}
\usepackage{braket}
\usepackage{color}
\usepackage{natbib}

\usepackage{multirow}
\usepackage{mathtools}
\usepackage{float}
\usepackage[caption = false]{subfig}
\usepackage{xcolor,colortbl}
\usepackage{color}

\newcommand{\be}{\begin{equation}}
\newcommand{\ee}{\end{equation}}
\newcommand{\ba}{\begin{eqnarray}}
\newcommand{\ea}{\end{eqnarray}}

\newcommand{\tr}{\operatorname{Tr}}

\newtheorem{theorem}{Theorem}

\newtheorem{definition}{Definition}

\newtheorem{observation}{Observation}

\newtheorem{remark}{Remark}
\newtheorem{lemma}{Lemma}

\def\>{\rangle}
\def\<{\langle}

\usepackage{stmaryrd}

\usepackage{centernot}
\usepackage{subfig}

\usepackage{diagbox}
\usepackage{multirow}
\usepackage{tabularx}

\begin{document}

\title{Gottesman–Knill Limit on One-way Communication Complexity: Tracing the Quantum Advantage down to Magic Resources}

\author{Snehasish Roy Chowdhury}
\affiliation{Physics and Applied Mathematics Unit, 203 B.T. Road Indian Statistical Institute Kolkata, 700 108, India}

\author{Sahil Gopalkrishna Naik}
\affiliation{S. N. Bose National Centre for Basic Sciences, Block JD, Sector III, Salt Lake, Kolkata 700 106, India}

\author{Ananya Chakraborty}
\affiliation{S. N. Bose National Centre for Basic Sciences, Block JD, Sector III, Salt Lake, Kolkata 700 106, India}

\author{Ram Krishna Patra}
\affiliation{S. N. Bose National Centre for Basic Sciences, Block JD, Sector III, Salt Lake, Kolkata 700 106, India}
\affiliation{HUN-REN Institute for Nuclear Research, P.O. Box 51, H-4001 Debrecen, Hungary.}

\author{Subhendu B. Ghosh}
\affiliation{S. N. Bose National Centre for Basic Sciences, Block JD, Sector III, Salt Lake, Kolkata 700 106, India}

\author{Pratik Ghosal}
\affiliation{S. N. Bose National Centre for Basic Sciences, Block JD, Sector III, Salt Lake, Kolkata 700 106, India}
\affiliation{Harish-Chandra Research Institute, A CI of Homi Bhabha National Institute, Chhatnag Road, Jhunsi, Allahabad - 211019, India.}

\author{Manik Banik}
\affiliation{S. N. Bose National Centre for Basic Sciences, Block JD, Sector III, Salt Lake, Kolkata 700 106, India}

\author{Ananda G. Maity}
\affiliation{S. N. Bose National Centre for Basic Sciences, Block JD, Sector III, Salt Lake, Kolkata 700 106, India}
\affiliation{Networked Quantum Devices Unit, Okinawa Institute of Science and Technology Graduate University, Onna-son, Okinawa 904-0495, Japan}
\affiliation{School of Physical Sciences, Indian Institute of Technology Goa, Ponda 403401, Goa, India}
	
\begin{abstract}
Quantum systems are known to offer advantages over their classical counterpart in communication complexity protocols, where the aim is to minimize amount of information exchange between distant parties to compute global functions of their distributed inputs. In this work, we establish that any one-way communication protocol implemented using a prime-dimensional quantum system---restricted to stabilizer-state encodings and Clifford-operation decodings---can be exactly simulated by transmitting a classical system of the same dimension, given access to shared randomness between the sender and receiver. In direct analogy with the Gottesman–Knill theorem, which attributes quantum computational speedup to non-stabilizer resources, commonly known as the magic resources, our result identifies the same non-stabilizer resources as the essential ingredient for the quantum advantage in one-way communication complexity. Furthermore, we present explicit tasks where even a `minimal magic resource' suffices to achieve a provable quantum advantage, highlighting its efficient use in communication protocols.
\end{abstract}

\maketitle	

\section{Introduction} The Second Quantum Revolution leverages non-classical features of quantum systems to devise technologies that surpass classical limitations \cite{Dowling2003,Deutsch20,Aspect2023}. Identifying the resources responsible for quantum advantage is crucial both for foundational insights and for guiding practical implementations. A striking illustration comes from the seminal Gottesman-Knill theorem, which shows that quantum circuits restricted to stabilizer state preparations and Clifford operations---despite creating large entanglement and superposition---can be efficiently simulated classically \cite{Gottesman1998}. However, inclusion of a single non-stabilizer resource, such as the T-gate, renders the system universal and confers quantum advantage, thereby identifying non-stabilizer resources as element of \textit{magic} in quantum computation \cite{Bravyi2005,Knill2005,Campbell2012,Bravyi2016,Howard2017,Bravyi19}. In this letter, we demonstrate that the same magic resource underpins the quantum advantage in communication complexity tasks. Specifically, for any prime dimension $d$, we prove that every one-way quantum protocol confined to stabilizer preparations and Clifford measurements always admits an exact simulation by a classical 
$d$-level system assisted with shared randomness (SR). Hence, non-stabilizer resources are necessary to exceed classical communication bounds.

\noindent Communication complexity, introduced by Yao \cite{Yao1979}, investigates the minimal information exchange required between distant parties to compute a joint function of their inputs (see also \cite{Yao1982,Kushilevitz1996,Kushilevitz1997}). In its simplest form, two parties, Alice and Bob, aim to compute a global function of their respective private inputs $x$ and $y$, drawn from finite alphabets $\mathcal{X}$ and $\mathcal{Y}$. To achieve this, the parties may exchange physical systems, subject to communication constraints, and are often allowed to access pre-shared classical correlation (also called SR) \cite{Ambainis2009,Bavarian2014,Canonne2015}. Quantum communication complexity generalizes this framework by introducing quantum resources. In Yao's quantum model \cite{Yao1993}, Alice is allowed to communicate quantum particles to Bob; whereas in the entanglement-assisted model of Cleve and Buhrman \cite{Cleve1997}, the parties share entanglement and communicate classically, gaining an advantage from the non-local correlations of the shared state \cite{Buhrman2010}. Variants like Random Access Codes (RACs) \cite{Wisener83,Ambainis2002} explore scenarios where quantum communication outperforms classical strategies for fixed communication costs. This quantum advantage is particularly striking in light of Holevo's no-go result \cite{Holevo1973}, which asserts that in standard Shannon-type communication scenarios, where no input is present at the receiver's end \cite{Shannon1948}, quantum communication offers no advantage over classical communication. More recently, Frenkel and Weiner extended this result by showing that, in the presence of SR, any input-output correlation achievable with a $d$-level quantum system can be reproduced using a $d$-level classical system \cite{Frenkel2015}.

\noindent These results naturally prompt the question: which specific quantum resources enable the observed advantage in communication complexity tasks? While previous studies partially answer this question highlighting the roles of coherent superposition in encoding states and incompatibility in decoding measurements \cite{Brukner04,Tavakoli2020,Carmeli2020,Guha2021,Gupta2024,Patra2024,Chakraborty2025,Chakraborty2025-1,Patra2026}, here we establish that the advantage fundamentally stems from non-stabilizer resources. Specifically, we show that in any one-way communication complexity scenario, when encodings are limited to stabilizer states and decodings to Clifford operations, all input-output correlations achievable with a prime-dimensional quantum system can always be reproduced by a classical system of the same dimension, given access to SR between the sender and receiver. Thus, the inclusion of magic resources, such as T-states or T-gates, is essential for surpassing classical limits. To this end, we also identify tasks that exhibit quantum advantage with minimal magic resource. We also identify the minimum amount of magic resources necessary to obtain an exponential advantage in communication complexity via a quantum strategy compared to its classical counterpart. 

\noindent The paper is organized as follows. In Sec.~\ref{sec2}, we briefly review the necessary preliminaries, including an overview of the stabilizer framework in Sec.~\ref{subsec2a} and one-way communication complexity in Sec.~\ref{subsec2b}. In Sec.~\ref{sec3}, we present our main results. In Sec.~\ref{subsec3a}, we prove that magic resources are necessary for any quantum advantage in one-way communication complexity. In Sec.~\ref{subsec3b}, we identify tasks for which a minimal amount of magic suffices to obtain a quantum advantage. In Sec.~\ref{subsec3c}, we establish conditions for robust exponential quantum advantages and determine the minimal magic required for quantum protocols to achieve exponentially lower communication costs than the best classical strategies. Finally, in Sec.~\ref{sec4}, we discuss broader implications of our results and outline directions for future research.
\section{Preliminaries}\label{sec2}
\subsection{Stabilizer framework}\label{subsec2a} Here we briefly revisit the stabilizer framework and for more detailed discussions we refer to the Refs.~\cite{Gottesman2002,Gross2006,Veitch2012,Veitch14}. For a general prime dimension $d$, the shift and phase operators are respectively given by: \(X \ket{j}= \ket{j\oplus_d1},~Z \ket{j} = \omega^j \ket{j}\), where \(\omega = \exp\left(2\pi \iota/d\right)\) with \(\iota=\sqrt{-1},~j\in\mathbb{Z}_d\), and \(\oplus_d\) denotes addition modulo \(d\). The generalized Pauli (Heisenberg-Weyl) operators are given by,
\begin{align}
\mathrm{P}_{(a_1, a_2)} = \begin{cases}
\iota^{a_1 a_2} ~X^{a_1} Z^{a_2};~a_1, a_2 \in \mathbb{Z}_2, \\
X^{a_1} Z^{a_2};~ a_1, a_2 \in \mathbb{Z}_d,~ d > 2.
\end{cases}   
\end{align}
The generalized Pauli group is comprised of the generalized Pauli operators along with a set of global phases: \(\mathcal{G}_2=\{\iota^{k} ~\mathrm{P}_{(a_1, a_2)}~|~  k\in \mathbb{Z}_4~\&~a_1,a_2 \in\mathbb{Z}_2\}\), and \(\mathcal{G}_d=\{\omega^{k} ~\mathrm{P}_{(a_1, a_2)}~|~  k,a_1,a_2 \in\mathbb{Z}_d\}\) for \(d>2\). The Clifford group \(\mathcal{C}_d\) consists of the unitaries under which the group \(\mathcal{G}_d\) remains closed, i.e. \(U \in \mathcal{C}_d\Leftrightarrow\forall~(a_1,a_2)~ \exists~(a'_1,a'_2)~\&~\phi~\text{s.t.}~
U~\mathrm{P}_{(a_1,a_2)} U^\dagger = \exp(\iota\phi) \mathrm{P}_{(a'_1,a'_2)}\). The set of stabilizer states for dimension $d$ are given by
\begin{align}
\mathrm{St}_d:=\text{ConvHull}\left\{U \ket{0}\bra{0}U^\dagger~|~U \in \mathcal{C}_d\right\}\subset\mathcal{D}(\mathbb{C}^d),\label{ext-state}
\end{align}
where \(\mathcal{D}(\cdot)\) is the set of density operators. Clearly, \(\mathrm{St}_d\) forms a polytope with its extreme points being the eigen-projectors of the mutually unbiased basis (MUB) operators \cite{Ivonovic1981}:
\begin{align}
\text{MUB}_d:=\left\{\mathrm{P}_{(0,1)},\mathrm{P}_{(1,0)},\mathrm{P}_{(1,1)},\cdots,\mathrm{P}_{(1,d-1)}\right\}.\label{MUB-d}
\end{align}
All possible combinations of the computational basis state preparations, computational basis measurements, and Clifford rotations form the set of stabilizer operations. In particular, this includes all stabilizer state preparations and measurements.
\begin{figure}[t]
\centering
\includegraphics[width=1\linewidth]{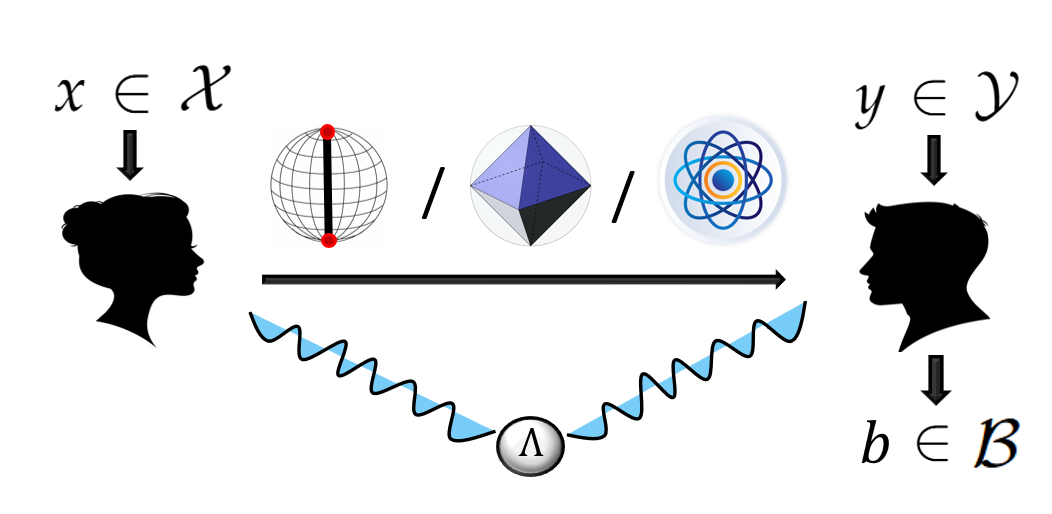}
\caption{One-way communication complexity setup. Given a classical input \(x\in\mathcal{X}\) to Alice and classical input \(y\in\mathcal{Y}\) to Bob, they collaboratively aims to reproduce a conditional probability distribution \(\{p(b|x,y)\}\) with minimizing the amount of communication from Alice to Bob, in assistance with SR \(\lambda\in\Lambda\). The sets \(\mathbb{S}_{\mathrm{K}}^d[\mathcal{X} \to \mathcal{B} \mid \mathcal{Y}]\) for \(\mathrm{K}=\mathrm{C},\mathrm{St},\mathrm{Q}\) respectively denote the correlations obtained through \(d\)-level classical, stabilizer-restricted quantum, and quantum communication from Alice to Bob.}\label{fig1}
\vspace{-.5cm}
\end{figure}

\subsection{One-way communication complexity}\label{subsec2b} Two distant parties, Alice and Bob, receive classical inputs $x \in \mathcal{X}$ and $y \in \mathcal{Y}$, respectively, and one of them (say Bob) aims to compute a function $f:\mathcal{X} \times \mathcal{Y} \to \mathcal{B}$ \cite{Yao1979}. A more general task involves simulating the conditional distribution $p({\mathcal{B}|\mathcal{XY}})\equiv\{p(b|x,y)\}$ while minimizing the amount of communication from Alice to Bob, possibly assisted by SR $\lambda \in \Lambda$ (see Fig.\ref{fig1}). In classical setting, their protocol goes as follow: Alice sends a message $m(x,\lambda)$ to Bob based on her input and the shared variable, whereas Bob produces an output $b = b(m,y,\lambda)$ that depends on the message received from Alice, his input \(y\), and the shared variable \(\lambda\). In the quantum case, Alice instead communicates a quantum state $\rho_{(x,\lambda)} \in \mathcal{D}(\mathbb{C}^d)$, which Bob measures using a POVM $M^y\equiv\{\pi^{k|y}\in\mathcal{E}(\mathbb{C}^d)~|~ \sum_k \pi^{k|y} = \mathbf{I}_d\}$, where \(\mathcal{E}(\cdot)\) is the set of effect operators and \(\mathbf{I}_d\) denotes the d-dimensional identity operator. He then computes the output $b = b(k,y,\lambda)$ based on the measurement outcome $k$, his input \(y\), and $\lambda$. The Canonical RAC tasks \cite{Wisener83,Ambainis2002} and their variants \cite{Spekkens2009,Banik2015,Perry2015,Tavakoli2015,Ambainis2019} demonstrate that quantum communication can outperform classical strategies in such one-way communication complexity tasks. In remainder of the paper we investigate the role of non-stabilizer resources in communication complexity.

\section{Results}\label{sec3}
\subsection{The Set up, Notations, and Definitions}
We denote by $\mathbb{S}_{\mathrm{C}}^d[\mathcal{X} \to \mathcal{B} \mid \mathcal{Y}]$, the set of input-output correlations $p(\mathcal{B} \mid \mathcal{X}\mathcal{Y})$ achievable via communicating a $d$-level classical system from Alice to Bob, while further assisted with SR. This set is convex for finite $\mathcal{X},~\mathcal{Y},$ and $\mathcal{B}$; and we reserve the notation $\mathbb{S}_{\mathrm{C}}^d[\mathcal{X} \to \mathcal{B}]$ for trivial $\mathcal{Y}$ scenario, i.e., when $|\mathcal{Y}| = 1$. The quantum counterpart, denoted $\mathbb{S}_{\mathrm{Q}}^d[\mathcal{X} \to \mathcal{B} \mid \mathcal{Y}]$, comprises correlations achievable through transmission of a $d$-dimensional quantum system. In RAC tasks, where Alice holds a random dit string ${\bf x}= x_1 \cdots x_n \in \mathcal{X}=\{0,1,\cdots d-1\}^n$, and Bob, upon receiving input $y \in \mathcal{Y} = \{1,\dots,n\}$, must guess $b = x_y \in \mathcal{B}=\{0,1,\cdots d-1\}$, quantum protocols are known to outperform classical ones \cite{Ambainis2002, Tavakoli2015}, thereby establishing the strict set inclusion $\mathbb{S}_{\mathrm{C}}^d[\mathcal{X} \to \mathcal{B} \mid \mathcal{Y}] \subsetneq \mathbb{S}_{\mathrm{Q}}^d[\mathcal{X} \to \mathcal{B} \mid \mathcal{Y}]$ for all $d \ge 2$. Apart from classical and quantum strategies, one can also define the class of Clifford-limited (or stabilizer-limited) strategies for one-way communication complexity tasks. 
\begin{definition}\label{def1}
[Stabilizer-limited strategy] A $d$-level stabilizer-limited strategy for one-way communication complexity restricts Alice's encoding to stabilizer states $\mathrm{St}_d$, and Bob’s decodings to stabilizer-preserving operations. 
\end{definition}
\noindent The classical, quantum, and stabilizer-limited strategies employed by Alice and Bob can further be categorized into three types: (i) pure (or extremal) strategies, where Alice and Bob follow definite encodings and decodings independently; (ii) mixed strategies, which invokes local randomness over pure strategies; and (iii) shared strategies, that further invokes SR between Alice and Bob, ensuring the respective sets of strategies and consequently the respective sets of simulable correlations to be convex. Before giving the formal definitions of theses strategies, we recall some features of general quantum measurement and the definition of extremal measurements.

\noindent  General quantum measurement, formally known as positive operator-valued measure (POVM) \cite{Kraus1983}, is defined as a collection of positive operators adding up-to the identity, i.e. \(M \equiv \{\pi_k~|~\pi_k\in\mathcal{E}(\mathbb{C}^d)~\&~\sum_k\pi_k=\mathbf{I}_d\}_{k=1}^l\). Set of effects forms a convex set, and the set of unnormalized effects forms a convex cone \(\mathcal{P}(\mathbb{C}^d)\subset\mathcal{L}(\mathbb{C}^d)\). An effect is said to be {\it ray extremal} if it can not be expressed as conical combination of other rays, i.e. \(\pi^{\text{ray-ext}}=p_1\pi^1+p_2\pi^2\) with \(p_1,p_2\ge0\) implies either \(\pi^1=\pi^2=\pi^{\text{ray-ext}}\) or one of the \(p_i\)'s is zero. Note that, all the extreme effects are not ray extremal. For instance, \(\mathbf{I}_d\) is extremal in the set of effects but is not ray extremal. In fact, ray extremal effects are exactly the rank-$1$ effects. One can further introduce the notion of extremal measurements. 
\begin{definition}\label{def2}
[Extremal POVM, Rank-$1$ extremal POVM] A $k$ outcome POVM \(\mathrm{M}\equiv\{\pi_a\}_{a=1}^k\) is called an extremal POVM if all the non-zero effects are extremal and linearly independent. Furthermore, it is called a rank-$1$ extremal POVM if all the non-zero effects are rank-1 (i.e. ray extremal) and linearly independent (Section {\bf 2.3.3} Ref.\cite{Watrous2018}). 
\end{definition}
\noindent Note that, all extremal POVMs are not rank-$1$ extremal. For instance, consider the qutrit measurement \(M_{0|12}=\{\ket{0}\bra{0},\ket{1}\bra{1}+\ket{2}\bra{2}\}\), which is an extremal POVM, but not rank-$1$ extremal. With this we can now formally define pure, mixed, and shared strategies for a generic one-way communication complexity task employed within classical, quantum, and stabilizer-limited scenarios. 
\begin{definition}{\label{def3}}
[Pure Strategies]

\noindent {\bf (a)} [Classical pure strategies] A \(d\)-level pure classical encoding-decoding strategy is an ordered tuple \((\mathbb{E}, D_1,\cdots, D_{|\mathcal{Y}|})\) that consists of an encoding function \(\mathbb{E}:\mathcal{X}\to\{0,1,\cdots,d-1\}\) and \(|\mathcal{Y}|\) decoding functions \(D_i:\{0,1,\cdots,d-1\}\to\mathcal{B}\) \cite{Ambainis2009}. 

\noindent {\bf (b)} [Quantum pure strategies] A \(d\)-level pure quantum encoding-decoding strategy is an ordered tuple \((\mathbb{P}_{\mathrm{Q}}, M^{\mathrm{Q}}_1,\cdots, M^{\mathrm{Q}}_{|\mathcal{Y}|},D_1,\cdots,D_{|\mathcal{Y}|})\) that consists of pure quantum state preparations \(\mathbb{P}_{\mathrm{Q}}:\mathcal{X}\to\{\ket{\psi_x}\}\subset\mathbb{C}^d\), \(|\mathcal{Y}|\) decoding measurements \(M^{\mathrm{Q}}_i:\{\ket{\psi_x}\}\to\mathcal{K}_i\), and \(|\mathcal{Y}|\) classical post-processing functions \(D_i:\mathcal{K}_i\to\mathcal{B}\), where \(M^{\mathrm{Q}}_i\)'s are quantum extremal POVMs (Definition \ref{def2}) with \(\mathcal{K}_i\) non-zero effects.

\noindent {\bf (c)} [Stabilizer-restricted pure strategies] A \(d\)-level pure stabilizer-restricted encoding-decoding strategy is an ordered tuple \((\mathbb{P}_{\mathrm{St}}, M^{\mathrm{St}}_1,\cdots, M^{\mathrm{St}}_{|\mathcal{Y}|},D_1,\cdots,D_{|\mathcal{Y}|})\) that consists of pure stabilizer state preparations \(\mathbb{P}_{\mathrm{St}}:\mathcal{X}\to\{\ket{\psi_x}\}\subset\mathrm{St}_d\), \(|\mathcal{Y}|\) decoding measurements \(M^{\mathrm{St}}_i:\{\ket{\psi_x}\}\to\mathcal{K}_i\), and \(|\mathcal{Y}|\) classical post-processing functions \(D_i:\mathcal{K}_i\to\mathcal{B}\), where \(M^{\mathrm{St}}_i\)'s are stabilizer preserving extremal POVMs with \(\mathcal{K}_i\) non-zero effects.
\end{definition}
\begin{definition}\label{def4}
[Mixed $\&$ Shared strategies] Such strategies invoke random variables \((\lambda_A,\lambda_B)\in\Lambda_A\times\Lambda_B\), sampled according to a probability distribution \(p(\lambda_A,\lambda_B)\), based on which Alice and Bob respectively randomize their pure encoding and decoding strategies. Whenever, the distribution is in product form, i.e. \(p(\lambda_A,\lambda_B)=p(\lambda_A)p(\lambda_B)\) it corresponds to a mixed strategy, else it is called a shared strategy.   
\end{definition} 

\noindent With these definitions in place, we now highlight a key observation specific to the extremal stabilizer-limited strategies.
\begin{observation}\label{obs1}
In an extremal strategy of a $d$-level stabilizer-limited protocol, Alice's encodings are restricted to the set $\mathrm{St}^{\mathrm{ext}}_d$ of $d(d+1)$ pure stabilizer states, while Bob's decodings are restricted to the $d+1$ sharp (projective) MUB measurements.
\end{observation}
\subsection{Magic Resources are Necessary for Quantum Advantage}\label{subsec3a}
\noindent Denoting the set of input-output correlations obtained through \(d\)-level stabilizer limited strategies along with SR as \(\mathbb{S}_{\mathrm{St}}^d[\mathcal{X} \to \mathcal{B} \mid \mathcal{Y}]\), we are now in a position to state one of the main contributions of the present work.  
\begin{theorem}\label{theo1}
For all the finite cardinality sets \(\mathcal{X}\), \( \mathcal{Y} \) and \(\mathcal{B}\) and for prime \(d\), we have \(\mathbb{S}_{\mathrm{St}}^d[\mathcal{X}\to\mathcal{B}|\mathcal{Y}]=\mathbb{S}_{\mathrm{C}}^d[\mathcal{X}\to\mathcal{B}|\mathcal{Y}]\).  
\end{theorem}
\begin{proof}
Any correlation in $\mathbb{S}_{\mathrm{St}}^d[\mathcal{X}\to\mathcal{B}|\mathcal{Y}]$ is either achieved by an extreme stabilizer strategy or by a shared strategy. Since, a shared strategy employs different extreme strategies based on some classical random variable shared between Alice and Bob (see Appendix for formal definition), to prove the theorem, it thus suffices to consider classical simulation of correlations achieved by extreme strategies only.   

In an extreme strategy, for each input $x \in \mathcal{X}$, Alice prepares and sends a pure stabilizer state $\psi_x\in \mathrm{St}^{\mathrm{ext}}_d$. Two types of encoding strategies arise--- (i) Uniform encoding: Alice uses a fixed stabilizer state $\psi_x$ for all $x \in \mathcal{X}$. This strategy clearly communicates no information about the input and is trivially classically simulable; (ii) Coarse-grained encoding: Alice partitions $\mathcal{X}$ into disjoint subsets \(\{\mathcal{X}_r\}_r\), i.e. \(\mathcal{X}=\cup_r\mathcal{X}_r~\&~\mathcal{X}_r\cap\mathcal{X}_{r'}=\emptyset\) for \(r\neq r'\). She then maps inputs of each subset to a distinct stabilizer state. We consider the maximal case where each state in $\mathrm{St}^{\mathrm{ext}}_d$ is used at least once, namely \(\mathcal{X}\) is partitioned into $d(d+1)$ disjoint subsets. Non-maximal cases with \(r<d(d+1)\) are discussed at the end. Similarly, Bob chooses a measurement from MUB\(_d\) based on his input $y \in \mathcal{Y}$. Again, we identify two cases: (i) Uniform decoding: Bob uses a fixed MUB measurement for all $y\in\mathcal{Y}$. As shown in \cite{Frenkel2015}, such strategies are classically simulable using a $d$-level classical system; (ii) Coarse-grained decoding: Bob partitions $y \in \mathcal{Y}$ into disjoint subsets \(\{\mathcal{Y}_t\}_t\), and performs different MUB measurements for different such subsets. The maximal case considers use of each of the MUB measurements at least once, i.e. $t=d+1$, whereas for non-maximally case \(t<d+1\).  Moreover, since each MUB measurement has $d$ outcomes, we may define an output alphabet $\mathcal{B}'$ with $|\mathcal{B}'| = d$. The original output alphabet $\mathcal{B}$ is connected to $\mathcal{B}'$ via classical postprocessing on Bob’s side.
\begin{small}
\begin{table}[t]
\centering
\footnotesize
\begin{tabular}{c||cc||cc||cc|}
& \multicolumn{2}{l||}{~~~$\mathcal{Y}_1\mapsto\mathrm{P}_{(0,1)}$}    & \multicolumn{2}{l||}{~~~$\mathcal{Y}_2\mapsto\mathrm{P}_{(1,0)}$}    & \multicolumn{2}{l|}{~~~$\mathcal{Y}_3\mapsto\mathrm{P}_{(1,1)}$}    \\ \hline
$\mathcal{X}_r\mapsto\psi^j_k$ &  \multicolumn{1}{l|}{$p(0|xy)$} & $p(1|xy)$ & \multicolumn{1}{l|}{$p(0|xy)$} & $p(1|xy)$ & \multicolumn{1}{l|}{$p(0|xy)$} & $p(1|xy)$ \\ \hline\hline
$\mathcal{X}_1\mapsto\psi^0_1$ & \multicolumn{1}{l|}{~~~~~$1$} & $0$ & \multicolumn{1}{l|}{~~~$1/2$} & $1/2$ & \multicolumn{1}{l|}{~~~$1/2$} & $1/2$  \\ \hline
$\mathcal{X}_2\mapsto\psi^1_1$ & \multicolumn{1}{l|}{~~~~~$0$} & $1$ & \multicolumn{1}{l|}{~~~$1/2$} & $1/2$ & \multicolumn{1}{l|}{~~~$1/2$} & $1/2$ \\ \hline
$\mathcal{X}_3\mapsto\psi^0_2$ & \multicolumn{1}{l|}{~~~$1/2$} & $1/2$ & \multicolumn{1}{l|}{~~~~~$1$} & $0$ & \multicolumn{1}{l|}{~~~$1/2$} & $1/2$ \\ \hline
$\mathcal{X}_4\mapsto\psi^1_2$ & \multicolumn{1}{l|}{~~~$1/2$} & $1/2$ & \multicolumn{1}{l|}{~~~~~$0$} & $1$ & \multicolumn{1}{l|}{~~~$1/2$} & $1/2$ \\ \hline
$\mathcal{X}_5\mapsto\psi^0_3$ & \multicolumn{1}{l|}{~~~$1/2$} & $1/2$ & \multicolumn{1}{l|}{~~~$1/2$} & $1/2$ & \multicolumn{1}{l|}{~~~~~$1$} & $0$ \\ \hline
$\mathcal{X}_6\mapsto\psi^1_3$ & \multicolumn{1}{l|}{~~~$1/2$} & $1/2$ & \multicolumn{1}{l|}{~~~$1/2$} & $1/2$ & \multicolumn{1}{l|}{~~~~~$0$} & $1$ \\ \hline
\end{tabular}
\caption{Input-output correlation \(p^{\max}({\mathcal{B'}|\mathcal{XY}})\). Alice communicates the pure stabilizer state \(\psi^j_k\) to Bob when her input belongs to the partition \(\mathcal{X}_r\), where \(r=2k+j-1\). Bob performs the MUB measurement \(\mathrm{P}_t\) if his input \(y\in\mathcal{Y}_t\), and obtains the outcome \(b\in\{0,1\}\).}
\label{taba1}
\end{table}\hspace{-.5cm}
\end{small}

Denoting \(\psi^j_k\) be the \(j^{\text{th}}\) eigenstate of the \(k^{\text{th}}\) MUB measurement \(\mathrm{P}_{k}\), with \(k=1\to(0,1),~\&~k\to(1,k-2)\) for \(k\ge2\), we have
\begin{align}\label{apb2}
\tr\left(\psi^j_k\psi^{j'}_{k'}\right)=\frac{1}{d}\left(1-\left(1-d\right)^{\delta_{jj'}}\delta_{kk'}\right).
\end{align}
In the maximal case, Alice partitions \(\mathcal{X}\) into \(d(d+1)\) disjoint sets \(\{\mathcal{X}_r\}_{r=1}^{d(d+1)}\), and encodes the input(s) in \(\mathcal{X}_r\) into the state \(\psi^j_k=\frac{1}{2}\left(\mathbf{I}_2+(-1)^j\mathrm{P}_{k}\right)\), where \(r=d(k-1)+j+1\) with \(j\in\{0,1,\cdots,d-1\},~k\in\{1,2,\cdots,d+1\}\). For the maximal case, \(t=d+1\), and Bob uses the measurements \(\{\mathrm{P}_{(0,1)},\mathrm{P}_{(1,0)},\mathrm{P}_{(1,1)},\cdots,\mathrm{P}_{(1,d-1)}\}\) measurements for his decoding. Without loss of any generality, considering \(\mathcal{B}'=\{0,1,\cdots,d-1\}\), the explicit form of the input-output correlation \(p^{\text{max}}({\mathcal{B'}|\mathcal{XY}})\) for \(d=2\) is provided in Table \ref{taba1} (for general \(d\) see Table \ref{taba2}). Notably, this correlation, within the stabilizer-limited qubit protocols, uses coherence in encoding and measurement incompatibility in  decoding, in a sense, in optimal way. 

We now show that the aforesaid correlation can be classically simulated using \(\log d\) bit of communication from Alice to Bob, along with the SR. We denote \(\log d\) bit of SR between Alice and Bob by \(\mathrm{R}(\lambda) = \tfrac{1}{d} \sum_{\lambda=0}^{d-1} \ket{\lambda}_A\bra{\lambda} \otimes \ket{\lambda}_B\bra{\lambda}\). To simulate this particular correlation, Alice and Bob share \(d+1\) bits of SR: \(\{\mathrm{R}(\lambda_k)\}_{k=1}^{d+1}\). Upon receiving the input \(x \in\mathcal{X}\), Alice decides in which partition \(\mathcal{X}_r\) it belongs to (where \(r=d(k-1)+j+1\)), and measures her share of the \(k^{\text{th}}\) shared bit. Upon obtaining the outcome \(\lambda_k\in\{0,1,\ldots,d-1\}\) she sends the message \(m = j \oplus_d(d- \lambda_k)\) to Bob, where \(\oplus_d\) is summation modulo \(d\). For the input \(y\in\mathcal{Y}\), Bob identifies the partition \(\mathcal{Y}_t\) in which it belongs, and measures his share of the \(t^{\text{th}}\) shared bit and outputs \(b = m \oplus_d \lambda_t\). A straightforward calculation confirms that this classical protocol exactly reproduces the correlation \(p^{\max}(\mathcal{B'}|\mathcal{X},\mathcal{Y})\). For the non-maximal case, some of the partitions \(\mathcal{X}_r\)'s become empty and/or \(t<d+1\); nevertheless, the same protocol remains applicable. The non-extremal correlations can be simulated with the aid of additional SR. This completes the proof.
\end{proof}
\begin{table*}[t]
\centering
\begin{small}
\renewcommand{\arraystretch}{1.6}
\setlength{\tabcolsep}{2.8pt}
\begin{tabular}{c|c|c|c|c|c|c|c|c|c|c|c|c|c|}
& \multicolumn{4}{c|}{$\mathcal{Y}_1 \mapsto \mathrm{P}_{(0,1)}$} 
& \multicolumn{4}{c|}{$\mathcal{Y}_2 \mapsto \mathrm{P}_{(1,0)}$} 
& $\cdots$
& \multicolumn{4}{c|}{$\mathcal{Y}_{d+1} \mapsto \mathrm{P}_{(1,d-1)}$} \\ \cline{2-14}
$\mathcal{X}_r \mapsto \psi_k^j$ 
& $p(0|xy)$ & $p(1|xy)$ & $\cdots$ & $p(d{-}1|xy)$ 
& $p(0|xy)$ & $p(1|xy)$ & $\cdots$ & $p(d{-}1|xy)$ 
& $\cdots$
& $p(0|xy)$ & $p(1|xy)$ & $\cdots$ & $p(d{-}1|xy)$ \\ \hline\hline

$\mathcal{X}_1 \mapsto \psi_1^0$ 
& 1 & 0 & $\cdots$ & 0 
& $\tfrac{1}{d}$ & $\tfrac{1}{d}$ & $\cdots$ & $\tfrac{1}{d}$ 
& $\cdots$
& $\tfrac{1}{d}$ & $\tfrac{1}{d}$ & $\cdots$ & $\tfrac{1}{d}$ \\ \hline

$\mathcal{X}_2 \mapsto \psi_1^1$ 
& 0 & 1 & $\cdots$ & 0 
& $\tfrac{1}{d}$ & $\tfrac{1}{d}$ & $\cdots$ & $\tfrac{1}{d}$ 
& $\cdots$
& $\tfrac{1}{d}$ & $\tfrac{1}{d}$ & $\cdots$ & $\tfrac{1}{d}$ \\ \hline

$\vdots$ 
& $\vdots$ & $\vdots$ & $\ddots$ & $\vdots$ 
& $\vdots$ & $\vdots$ & $\ddots$ & $\vdots$
& $\ddots$
& $\vdots$ & $\vdots$ & $\ddots$ & $\vdots$ \\ \hline

$\mathcal{X}_d \mapsto \psi^{d-1}_1$ 
& 0 & 0 & $\cdots$ & 1 
& $\tfrac{1}{d}$ & $\tfrac{1}{d}$ & $\cdots$ & $\tfrac{1}{d}$ 
& $\cdots$
& $\tfrac{1}{d}$ & $\tfrac{1}{d}$ & $\cdots$ & $\tfrac{1}{d}$ \\ \hline

$\mathcal{X}_{d+1} \mapsto \psi_2^0$ 
& $\tfrac{1}{d}$ & $\tfrac{1}{d}$ & $\cdots$ & $\tfrac{1}{d}$  
& 1 & 0 & $\cdots$ & 0 
& $\cdots$
&$\tfrac{1}{d}$ & $\tfrac{1}{d}$ & $\cdots$ & $\tfrac{1}{d}$ \\ \hline

$\vdots$ 
& $\vdots$ & $\vdots$ & $\ddots$ & $\vdots$ 
& $\vdots$ & $\vdots$ & $\ddots$ & $\vdots$ 
& $\ddots$
& $\vdots$ & $\vdots$ & $\ddots$ & $\vdots$ \\ \hline

$\mathcal{X}_{d(d+1)} \mapsto \psi_{d+1}^{d-1}$ 
& $\tfrac{1}{d}$ & $\tfrac{1}{d}$ & $\cdots$ & $\tfrac{1}{d}$ 
& $\tfrac{1}{d}$ & $\tfrac{1}{d}$ & $\cdots$ & $\tfrac{1}{d}$ 
& $\cdots$
& 0 & 0 & $\cdots$ & 1\\ \hline 
\end{tabular}
\end{small}
\caption{The input-output correlation \(p^{\max}({\mathcal{B'}|\mathcal{XY}})=\left\{p(b|x,y)\right\}\) for arbitrary prime \(d\). Alice communicates the pure stabilizer state \(\psi^j_k\) to Bob when her input belongs to the partition \(\mathcal{X}_r\), where \(r=d(k-1)+j+1\). Bob performs the MUB measurement \(\mathrm{P}_t\) if his input \(y\in\mathcal{Y}_t\), and obtains the outcome \(b\in\{0,1,\cdots,d-1\}\).}
\label{taba2}
\end{table*}

\noindent At this point it is worth recalling the result of Frenkel and Weiner, which establishes that $\mathbb{S}_{\mathrm{C}}^d[\mathcal{X}\to\mathcal{B}]= \mathbb{S}_{\mathrm{Q}}^d[\mathcal{X}\to\mathcal{B}]$ for all finite $\mathcal{X}$ and $\mathcal{B}$ \cite{Frenkel2015}. That is, in the absence of input at Bob’s end and with SR, any input-output correlation achievable with a $d$-dimensional quantum system can also be realized with a $d$-dimensional classical system (see also \cite{DallArno2017,Naik2022,Patra2023,Patra2024} for other related works). Theorem \ref{theo1} generalizes this no-go result to the communication complexity setting. Analogous to the Gottesman–Knill theorem,  which shows that quantum circuits composed of Clifford gates and restricted to Clifford operations are classically simulable, Theorem \ref{theo1} demonstrates that stabilizer-limited quantum protocols cannot yield any advantage over classical ones in communication complexity tasks when SR is available between the communicating parties.
\begin{figure}[t]
\centering
\includegraphics[width=1\linewidth]{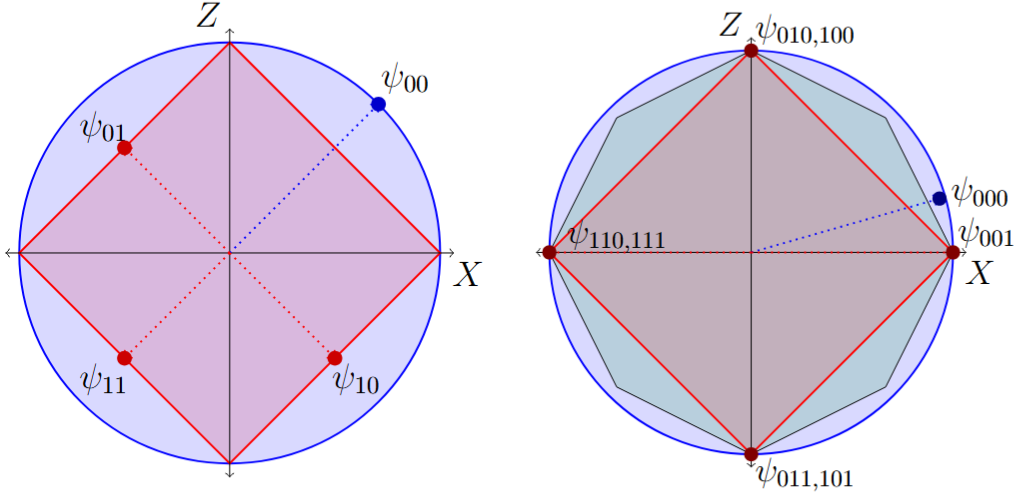}
\caption{(Color online) The XZ-plane of the Bloch sphere is shown. Stabilizer states lie within the red shaded square, while all other states in the plane are non-stabilizer. [{\bf Left}] In $2 \mapsto 1$ RAC, a quantum advantage is achieved when at least one encoding state (shown as a blue dot) lies outside the stabilizer set, even if all other encoding states (red dots) are stabilizer states. The specific encoding-decoding configuration illustrated here achieves a success probability of $\frac{1}{16}(11 + \sqrt{2})$. [{\bf Right}] In $3 \mapsto 1$ RAC, a quantum advantage is obtained only if the non-stabilizer encoding lies outside the green shaded octagonal region, with all other encodings chosen from the stabilizer set. This octagon is defined by the inequalities $|2n^x| + |n^z| \le 2$ and $|n^x| + |2n^z| \le 2$. Similar geometric constraints apply to the $2 \mapsto 1$ RAC task when non-stabilizer encodings outside the XZ-plane are considered, supporting the claim in Remark \ref{remark1}.
 }\label{fig2}
\end{figure}

Notably, Theorem \ref{theo1} hinges on the structure of mutually unbiased bases (MUBs) \cite{McNulty2024}. It holds whenever preparations and measurements are restricted to the convex hull of eigenstates of a chosen MUB, but fails for non-MUB sets. For example, states forming a regular octagon on a Bloch-sphere great circle reproduce the optimal qubit $2 \mapsto 1$ random access code, which requires more than one classical bit. On the other hand, when $d$ is a prime power, the tensor product of Heisenberg-Weyl operators readily gives $(d+1)$ MUBs, which is known to be optimal \cite{Bandyopadhyay2002}. Moreover, the algebraic structure of these constructions also identifies the set of free transformations underlying the generalized stabilizer sub-theory of the Gottesman–Knill theorem. However, for arbitrary $d$ (e.g. $d=6,10$), such an elegant construction of MUBs with the set of free transformation is elusive, even-though one can come up with set of sub-optimal MUBs. While our Theorem holds for such suboptimal MUBs, but without specifying the set of free transformation, the corresponding sub-theory lacks the physical motivation.
\subsection{Minimal Magic Resource for Quantum Advantage}\label{subsec3b}
\noindent A natural question is whether the presence of magic resources always guarantees a quantum advantage in communication complexity tasks. Our next result identifies a scenario where even a `minimal magic resource' suffices to achieve such an advantage. Before presenting our result, we note that various measures have been proposed to quantify magic in non-stabilizer states \cite{Veitch14,Veitch2012,Wang2020,Koukoulekidis2022}. Crucially, our result holds for any faithful measure. 
\begin{lemma}\label{theo2}
Arbitrarily small nonzero magic in any one of the encoding states ensures quantum advantage in the $2 \mapsto 1$ RAC task, even when all other encoding states are restricted to stabilizer set and all the decodings are limited to Clifford operations.
\end{lemma}
\begin{proof}
For the \(2\)-level \(2\mapsto1\) RAC, let Alice encode her bit string \({\bf x}=x_1x_2\in\{0,1\}^2\) on qubit state \(\psi_{{\bf x}}=\frac{1}{2}(\mathbf{I}_2+\vec{n}_{{\bf x}}\cdot\sigma)\), where \(\vec{n}_{{\bf x}}=(n^x_{\bf x},n^y_{\bf x},n^z_{\bf x})^{\mathrm{T}}\in\mathbb{R}^3\) denotes the Bloch vector of the encoding state. Let for the input \(y\in\{1,2\}\), Bob performs the stabilizer measurements \(\mathrm{P}_{y}\equiv\{\pi^0_{y},\pi^1_{y}\}\), where \(\pi^b_{y}:=\frac{1}{2}(\mathbf{I}_2+(-1)^b\mathrm{P}_y)\), with \(y=1\equiv(0,1)\), \(y=2\equiv(1,0)\), and \(b\in\{0,1\}\). This protocol yields a success \(P_{succ}:=\frac{1}{8}\sum_{{\bf x},y}p(b=x_y|{\bf x},y)=\frac{1}{8}\sum_{{\bf x},y}\tr[\psi_{\bf x}\pi^b_y]\), which explicitly reads as
\begin{align}
P_{succ}=\frac{1}{16}\left[8+\hspace{-.1cm}\sum_{x_1,x_2}\hspace{-.1cm}\left\{(-1)^{x_1}n^z_{\bf x}+(-1)^{x_2}n^x_{\bf x}\right\}\right].
\end{align}
When all encoding states are restricted to the stabilizer set $\mathrm{St}_2$, each term in curly brace within the summation attains a maximum value of $1$, leading to an optimal success probability of $P^{\mathrm{St}}_{\text{succ}} = 3/4$, which coincides with the optimal classical value as consistent with Theorem \ref{theo1}. Introducing a nonzero amount of magic $\epsilon > 0$ in only one encoding state increases the success probability to $P^{\epsilon}_{\text{succ}} = (12 + \epsilon)/16$. This completes the proof.
\end{proof}
An analogous result also hold for $3 \mapsto 1$ which we formalize in the following lemma,
\begin{lemma}\label{theo2E}
Arbitrarily small nonzero magic in any one of the encoding states ensures quantum advantage in the $3 \mapsto 1$ RAC task, even when all other encoding states are restricted to stabilizer set and all decodings are limited to Clifford operations. 
\end{lemma}
\begin{proof}
In $3 \mapsto 1$ RAC, Alice receives a random $3$-bit string $\textbf{x}=x_1x_2x_3\in\{0,1\}^3$, and Bob has to output $b=x_y\in\{0,1\}$ based on his random input $y\in\{1,2,3\}$. Likewise, \(2\mapsto1\) case, let Alice encode her inputs in qubit states of the form $\textbf{x}\to\psi_{\textbf{x}}=\frac{1}{2}(\mathbf{I}_2+\vec{n}_{\textbf{x}}\cdot\sigma)$. To optimize the success probability, Alice's encodings must be chosen properly depending on Bob's decoding strategies. Within Clifford-restricted scenario, we analyses three three different cases for Bob decoding.

{\bf I.} Irrespective of the input \(y\), Bob always performs a fixed decoding measurement (say \(\mathrm{P}_{(1,0)}\)). However, this strategy---when comprised with any encoding on Alice's side---cannot surpass the optimal classical success probability of \(3/4\). This limitation stems from the fact that achieving a quantum advantage in RAC tasks requires measurement incompatibility at the decoding stage \cite{Carmeli2020}, a feature absent in the above strategy. This also follows from the result of Frenkel and Weiner \cite{Frenkel2015}. Nonetheless, the encoding \(\textbf{x}\to\psi_{\textbf{x}}=\frac{1}{2}(\mathbf{I}_2+(-1)^{\text{Maj}(\textbf{x})} \mathrm{P}_{(1,0)})\) achieves the success probability \(3/4\), here \(\text{Maj}(\cdot)\) denotes the majority function \cite{Ambainis2024}.

{\bf II.} Bob performs a three different MUBs for three of his different inputs: 
\begin{align}
y\to &\mathrm{P}_{y}\equiv\{(\mathbf{I}_2+(-1)^b\mathrm{P}_{y})/2\}_{b=0}^1,~~y\in\{1,2,3\}\nonumber\\
&\text{with}~1\equiv(0,1),~2\equiv(1,0),~3\equiv(1,1).
\end{align}   
This protocol yields a success  
\small
\begin{align}
P_{succ}=\frac{1}{48}[24+\hspace{-.2cm}\sum_{x_1,x_2,x_3}\hspace{-.2cm}\{(-1)^{x_1}n^z_{\bf x}+(-1)^{x_2}n^x_{\bf x}+(-1)^{x_3}n^y_{\bf x}\}].
\end{align}
\normalsize
When all encoding states are restricted to the stabilizer set $\mathrm{St}_2$, each term in curly brace within the summation attains a maximum value of $1$, leading to an average success probability of $P_{\text{succ}} = 2/3$, a value strictly less than the classical optimal success \(3/4\). Furthermore, allowing magic in one state only, one can achieve a success almost \(\frac{1}{48}(31+\sqrt{3})<\frac{3}{4}\). With this encodings, an advantage over the classical protocols requires at least six encoding states to be non-stabilizer. However, magic in all the eight encoding states can achieve the optimal quantum success \(\frac{1}{2}(1+\frac{1}{\sqrt{3}})\). 

Thus, neither the strategy (a) nor (b) establishes our claim. We thus consider a third strategy as below.

{\bf III.} Bob performs one particular MUB measurement for two of the inputs and a different measurement for the remaining input, i.e., 
\begin{subequations}
\begin{align}
y\in\{1,2\}&\to\mathrm{P}_{(1,0)}\equiv\{(\mathbf{I}_2+(-1)^b \mathrm{P}_{(1,0)})/2\}_{b=0}^1,\\
y=3&\to\mathrm{P}_{(0,1)}\equiv\{(\mathbf{I}_2+(-1)^b \mathrm{P}_{(0,1)})/2\}_{b=0}^1.
\end{align}
\end{subequations}
In this case, the success probability reads as
\small
\begin{align}
&P_{succ}=\frac{1}{48}[24+\hspace{-.3cm}\sum_{x_1,x_2,x_3}\hspace{-.3cm}\{[(-1)^{x_1}+(-1)^{x_2}]n^x_{\bf x}+(-1)^{x_3}n^z_{\bf x}\}]\nonumber\\
&=\frac{1}{48}[24+\hspace{-.3cm}\sum_{x_1=x_2,x_3}\hspace{-.3cm}\{(-1)^{x_1}2n^x_{\bf x}+(-1)^{x_3}n^z_{\bf x}\}+\hspace{-.3cm}\sum_{x_1\neq x_2,x_3}\hspace{-.3cm}(-1)^{x_3}n^z_{\bf x}].  
\end{align}
\normalsize
When, all the encodings are restricted to \(\mathrm{St}_2\), each of the terms within curly braces in first summation can achieve maximum value \(2\), whereas each of the terms in second summation can achieve maximum value \(1\), thereby leading to the optimal classical success \(3/4\). If one of the encoding states contributing in the first summation is allowed to be a non-stabilizer state, then it is possible to surpass the classical optimal success (See Fig. \ref{fig2}). This completes the proof. 
\end{proof}

It is noteworthy and somewhat counterintuitive that in Case \textbf{II}, where Alice employs quantum but stabilizer-restricted encoding and Bob performs three distinct MUB measurements, the average success probability falls below that of the optimal classical strategy. The classical optimum is achieved by \emph{majority encoding with identity decoding} \cite{Ambainis2024}, where Bob’s action is input-independent and performance is determined solely by Alice’s encoding. By encoding the majority bit of each 3-bit string, Alice achieves success probability $2/3$ for six strings and $1$ for `$000$' and `$111$'. In contrast, the quantum strategy with three MUBs ($X,Y,Z$) requires input-dependent measurements by Bob: any pure encoding yields certainty in one basis but only random outcomes ($1/2$) in the other two, resulting in an average success probability of $2/3$ for all eight strings. Thus, the MUB-based quantum strategy performs strictly worse than the classical optimum, a limitation arising from the intrinsic randomness imposed by mutual unbiased-ness.

At this point a crucial difference in Lemma \ref{theo2} and Lemma \ref{theo2E} is worth mentioning. Notably in case of $2\mapsto 1$ RAC, by considering only one magic state for encoding one can obtain the optimal success $\frac{1}{16}(11 + \sqrt{2})$ [the encodings are depicted in Fig. \ref{fig2}], while unrestricted encodings achieve the optimal quantum success of $\frac{1}{2}(1+\frac{1}{\sqrt{2}})$. Strikingly, in $3 \mapsto 1$ RAC, only one magic state encoding does not yield any advantage when three distinct non-commuting stabilizer-preserving measurements are used for Bob’s three inputs. In contrast, an advantage emerges when only two such measurements are employed, with one reused for two different inputs. Further to note that, Lemma \ref{theo2} and \ref{theo2E}, as specified in the following remark, requires a careful reading.
\begin{remark}\label{remark1}
Although Lemma \ref{theo2} and Lemma \ref{theo2E} establish that only one magic state encoding can lead to an advantage in RACs, not all of the non-stabilizer states necessarily yield such an advantage (see Fig.\ref{fig2}).    
\end{remark}

We can generalize Lemma \ref{theo2} and Lemma \ref{theo2E} for arbitrary 2-level $N\mapsto1$ RAC tasks as depicted in the following theorem.
\begin{theorem}\label{theo2s}
Arbitrarily small nonzero magic in any one of the encoding states ensures quantum advantage in all $N \mapsto 1$ RAC tasks for $N\in\mathbb{N}$, even when all other encoding states are restricted to stabilizer set and all decodings are limited to Clifford operations. 
\end{theorem}
\begin{proof}
Let Alice's input string be $\textbf{x}=x_1x_2\cdots x_N\in\{0,1\}^N$ and Bob's input is $y\in\{1,2,\cdots,N\}$. We analyze the two cases $N=\text{odd}$ and $N=\text{even}$ separately.  \\

\noindent \textbf{Case I ($N=\text{odd}$):}
Let $N=2n+1$, where $n\in\mathbb{N}$. Alice considers her encoding strategy as follows: she first check the majority (`Maj') of the first $2n$ bits. In case of unique majority she encodes the string $\textbf{x}$ into the state 
\begin{align}\label{maj}
\textbf{x}\to\psi_{\textbf{x}}=\frac{1}{2}\left(\mathbf{I}_2+(-1)^{\eta} P_{(0,1)}\right),
\end{align}
where $\eta=\text{Maj}(x_1x_2\cdots x_{2n})\in\{0,1\}$. If there is no unique majority, then her encoding depends on the $(2n+1){th}$ bit as follows 
\begin{align}\textbf{x}\to\psi_{\textbf{x}}=\frac{1}{2}\left(\mathbf{I}_2+(-1)^{x_{(2n+1)}} P_{(1,0)}\right).
\end{align}
For a better clarity of this encoding strategy, consider $n=2$. Then Alice's encodings read as:
\begin{align}
\left.\begin{aligned}
&00000\to\ket{0}\bra{0},\quad00111\to\ket{-}\bra{-},\\
&00001\to\ket{0}\bra{0},\quad 11110\to\ket{1}\bra{1},\\
&\hspace{1.2cm} \vdots\hspace{3.8cm} \vdots\\
&00110\to\ket{+}\bra{+},~11111\to\ket{1}\bra{1}
\end{aligned}\right\}    
\end{align}
Bob employs the decoding strategy $M_y\equiv\{\pi^b_y\}_{b=0}^1$, with 
\begin{subequations}
\begin{align}
&\pi^b_y:=\frac{1}{2}\left(\mathbf{I}_2+(-1)^b P_{(0,1)}\right),~\text{for}~y\in\{1,\cdots,2n\},\\
&\pi^b_y:=\frac{1}{2}\left(\mathbf{I}_2+(-1)^b P_{(1,0)}\right),~~\text{for}~~y=(2n+1).
\end{align}\label{decod}
\end{subequations}
Recall that the optimal classical success of $N\mapsto1$ RAC with $1$-bit of communication is achieved with the `majority encoding and identity decoding' (MEID) strategy \cite{Ambainis2024}. We will now argue that the aforementioned quantum strategy, which from now on we will refer to as the odd non-magic quantum (ONMQ) strategy, also achieve the same success.
\begin{itemize}
\item Consider the strings that have a unique majority in first $2n$ bits, with exactly $(n+k)$ bits having the value $0(1)$, where $1\le k\le n$. For such strings, success with classical MEID strategy and ONMQ strategy is listed in Table \ref{tab1}.
\begin{table}[ht]
\centering
\begin{tabular}{c||c|c|}
$x_{2n+1}$& MEID & ONMQ\\\hline
$1(0)$ & $(n+k)/(2n+1)$ &$(n+k+1/2)/(2n+1)$\\\hline
$0(1)$& ~~$(n+k+1)/(2n+1)$~~ & ~~$(n+k+1/2)/(2n+1)$~~\\\hline
\end{tabular}
\caption{Comparison of success probability of $(2n+1)\mapsto1$ RAC under classical MEID and ONMQ strategies, when Stings having unique majority in first $2n$ bits.}
\label{tab1}
\end{table}

\item Consider the strings where same number of $0$'s and $1$'s are there in first $2n$ bits. For all such strings, both the MEID strategy and the ONMQ strategy yield a success probability $P_{Succ}:=(n+1)/(2n+1)$
\end{itemize}
It thus turns out that the average success probability with theses two strategies will be same. However, the ONMQ strategy now helps us to argue that by allowing magic into one of the encoding states we can surpass the optimal classical success. For that, let one of the string (say $\textbf{x}^\star$) having $(n+k)$ number of $0$'s is encoded in quantum state 
\begin{align}
\phi_{\textbf{x}^\star}=\frac{1}{2}(\mathbb{I}_2+\vec{r}.\sigma),~\text{where}~|\vec{r}|\le1. 
\end{align}
The decoding in Eq.(\ref{decod}), for this particular string, yields a success
\begin{align}
P_{\textbf{x}^\star}&=\frac{1}{2}(1+r_z)(n+k)+\frac{1}{2}(1-r_z)(n-k)\nonumber\\  &\quad ~~+\frac{1}{2}(1+(-1)^{x^{\star}_{2n+1}}r_x)\nonumber\\
&=\frac{1}{2}+n+2k r_z+(-1)^{x^{\star}_{2n+1}}r_x.
\end{align}
In ONMQ encoding we have chosen $r_z=1$ and $r_x=0$ leading to a success
\begin{align}
P^{ONMQ}_{\textbf{x}^\star}=\frac{1}{2}+n+2k.   
\end{align}
Now, for all $k>1$ one can choose a suitable magic state $\phi_{\textbf{x}^\star}$ such that $2k r_z+(-1)^{x^{\star}_{2n+1}}r_x>2k$, thereby surpassing the optimal classical success.  

\noindent\textbf{Case II ($N=\text{even}$):} In case of $N=2n$, where $n\in\mathbb{N}$, Alice considers her encoding strategy as follows: she first check the majority (`Maj') of the first $2n-1$ bits. Let $\eta=\text{Maj}(x_1x_2\cdots x_{2n-1})\in\{0,1\}$, then she encodes the string $\textbf{x}$ into the state 
\begin{align}\label{majeven}
\textbf{x}\to\psi_{\textbf{x}}=\frac{1}{2}\left(\mathbf{I}_2+(-1)^{\eta} P_{(0,1)}\right).
\end{align}
We will consider the encoding strategy for $n=2$ as an example of Alice's encoding:
\begin{align}
\left.\begin{aligned}
&0000\to\ket{0}\bra{0},\quad1000\to\ket{0}\bra{0},\\
&0001\to\ket{0}\bra{0},\quad 1001\to\ket{0}\bra{0},\\
&\hspace{1.2cm} \vdots\hspace{3.2cm} \vdots\\
&0111\to\ket{1}\bra{1},~~~~1111\to\ket{1}\bra{1}
\end{aligned}\right\}    
\end{align}
Bob employs the same decoding strategy of odd $N$ (Eq.(\ref{decod})) $M_y\equiv\{\pi^b_y\}_{b=0}^1$, with 
\begin{subequations}
\begin{align}
&\pi^b_y:=\frac{1}{2}\left(\mathbf{I}_2+(-1)^b P_{(0,1)}\right),\\&~~\qquad\qquad~~\text{for}~~y\in\{1,2,\cdots,(2n-1)\}\nonumber,\\
\text{and}~~~&\pi^b_y:=\frac{1}{2}\left(\mathbf{I}_2+(-1)^b P_{(1,0)}\right),~~\text{for}~~y=2n.
\end{align}\label{decodeven}
\end{subequations}
This quantum encoding and decoding strategy from now on will be referred as even non-magic quantum strategy (ENMQ). We will establish that ENMQ and classical optimal strategy majority encoding identity decoding (MEID) will give the same success probability.
Lets consider a string of length $2n$, with exactly $(n+k)$  bits of first $(2n-1)$ bits have the value $0(1)$, where $0\leq k\leq n-1$. For such strings, the success with classical MEID strategy and ENMQ is listed below in Table \ref{tab2}.
\begin{table}[ht]
\centering
\begin{tabular}{c||c|c|}
$x_{2n+1}$& MEID & ENMQ\\\hline
$\bar{\eta}$ & $(n+k)/2n$ &$(n+k+1/2)/2n$\\\hline
$\eta$& ~~$(n+k+1)/2n$~~ & ~~$(n+k+1/2)/2n$~~\\\hline
\end{tabular}
\caption{Comparison of success probability of $2n$ bit string under classical MEID and ENMQ strategies, where $\eta=\text{Maj}(x_1,x_2,\cdots,x_{2n})\in\{0,1\}$.}
\label{tab2}
\end{table}
As the number of $2n$ bitstrings with $x_{2n} = \eta$ exactly matches the number of bitstrings with $x_{2n} \neq \eta$, the average success probability in both strategies turns out to be the same. However, the NMQ strategy now helps us to argue that by allowing magic into one of the encoding states, we can surpass the optimal classical success. For that, let one of the string (say $\textbf{x}^\star$) having $(n+k)$ number of $0$'s in the first $2n-1$ bit is encoded in quantum state 
\begin{align}
\phi_{\textbf{x}^\star}=\frac{1}{2}(\mathbb{I}_2+\vec{r}.\sigma),~\text{where}~|\vec{r}|\le1. 
\end{align}
The decoding in Eq.(\ref{decodeven}), for this particular string, yields a success
\begin{align}
P_{\textbf{x}^\star}&=\frac{1}{2}(1+r_z)(n+k)+\frac{1}{2}(1-r_z)(n-k-1)\nonumber\\&\qquad~~+\frac{1}{2}(1+(-1)^{x^{\star}_{2n}}r_x)\nonumber\\
&=n+(2k+1) r_z+(-1)^{x^{\star}_{2n}}r_x.
\end{align}
In ENMQ encoding we have chosen $r_z=1$ and $r_x=0$ leading to a success
\begin{align}
P^{ENMQ}_{\textbf{x}^\star}=n+2k+1.   
\end{align}
Now, for all $k\geq 0$ one can choose a suitable magic state $\phi_{\textbf{x}^\star}$ such that $(2k+1) r_z+(-1)^{x^{\star}_{2n+1}}r_x>2k+1$, thereby surpassing the optimal classical success.

\noindent This completes the proof.
\end{proof}
\subsection{Magic Resources and Exponential Quantum Advantage}\label{subsec3c}
In the preceding sections, we established that, in prime dimensions, stabilizer resources offer no advantage over classical strategies in one-way communication complexity. Interestingly, it is known that quantum resources can be exponentially advantageous over the classical resources in one-way communication complexity \cite{Raz1999, BarYossef2004, BarYossef2008, Gavinsky2006, Havlek2020}. Our Theorem~\ref{theo1} implies that quantum strategies exhibiting such an exponential advantage must employ the magic resource in the protocol. Naturally the question arises how much magic resources are indeed required to achieve such an exponential advantage. In what follows, we quantify the minimum amount of magic required to realize this advantage. To this end, we begin by formally defining the operational setting in which exponential quantum advantage is assessed.

\begin{definition} \label{def5}
[Communication Problem] A communication problem (or simply problem) $\mathcal{P}$ is defined by an infinite family of relations $\mathcal{P} := \{R_n\}_{n \ge 1}$, where $R_n \subseteq X_n \times Y_n \times Z_n$. $X_n = \{0, 1\}^n$ and $Y_n = \{0, 1\}^n$ are the input sets for Alice and Bob and $Z_n$ is the set of possible outputs. A pair $(x, y) \in X_n \times Y_n$ is a valid input if there exists at least one $z \in Z_n$ such that $(x, y, z) \in R_n$. If all $(x, y) \in X_n \times Y_n$ is a valid input for all $n\in \mathbb{N}$, then the problem $\mathcal{P}$ is called a `total' problem otherwise it is called a `promise' problem.
\end{definition}
We say, a problem $\mathcal{P}$ is solvable within an error $\epsilon< 1/2$ by a quantum (classical) shared strategy via communicating $\mathcal{Q}_\epsilon(n)$ ($\mathcal{C}_\epsilon(n)$) qudits (cdits) if for all $n$ and all valid inputs $(x, y)$, the output $b$ of the protocol satisfies: 
\begin{align*}
\sum_{b: (x,y,b)\in R_n}p(b| x,y)&\ge 1-\epsilon
\end{align*}
With this, we can now define exponential quantum advantage for a problem in the following way:
\begin{definition}\label{def6}
    [Exponential Quantum Advantage] A communication problem $\mathcal{P}$ is said to exhibit exponential quantum advantage if there exist one quantum shared strategy requiring $\mathcal{Q}_{\epsilon}(n)$ amount of qudit communication to solve $\mathcal{P}$ within an error $\epsilon$ such that
    \[
    \mathcal{C}_{\epsilon}(n) \ge d^{\Omega(\mathcal{Q}_{\epsilon}(n))}
    \]
    Where, $\mathcal{C}_{\epsilon}(n)$ is the amount of cdit communication required for the optimal classical shared strategy to solve the problem $\mathcal{P}$ within an error $\epsilon$.
\end{definition}
Our question of interest is how much magic resources are needed in a quantum strategy that achieve an exponential advantage in any such problem. We start by noting down following two important observations: 
\begin{observation}\label{obs2}
Any quantum strategy that admits an exponential advantage in some communication problem, must use some magic states in the encoding step, even when there in no limit on the amount of magic used during the decoding step.  
\end{observation}
This observation follows from the Gottesman--Knill theorem \cite{Gottesman1998} (see also \cite{karanjai2018contextuality}) together with a defining structural feature of the communication complexity model, namely that the communicating parties are endowed with unbounded local computational power. In this setting, the operational cost of a simulation is quantified solely by the classical channel capacity required to transmit the instructions specifying the protocol, and not by the time or computational resources needed to reproduce quantum statistics locally.

In particular, when the encoding stage is restricted to stabilizer states, the entire quantum protocol admits an efficient classical emulation. The set of extreme stabilizer states on $\mathcal{Q}(n)$ qudits can be specified using only $\Omega(\mathcal{Q}(n)^2)$ classical dits, implying that the quantum communication can be replaced by a classical message of quadratic size. Upon receiving this classical description, Bob can deterministically reproduce the operational predictions of the protocol by using his unbounded local computational power to analytically evaluating the Born rule probabilities $p(b|x,y)=\operatorname{Tr}(\Pi_y \rho_x)$. Thus, in the absence of non--stabilizer resources, the correlations generated by the protocol can be simulated by a purely classical strategy with at most polynomial overhead in communication.
\begin{observation}\label{obs3}
Usage of magic resources are not always necessary at the decoding step to ensure an exponential quantum advantage in communication problem.
\end{observation}
The claim follows through an explicit example. Particularly, the well known Hidden Matching problem \cite{BarYossef2004, BarYossef2008}, which exhibits such exponential advantage, employs only stabilizer resources at the decoding step. Our next theorem provides a lower bound on the amount of magic resources needed at the encoding step to ensure an exponential advantage. 
\begin{theorem}\label{theorem3}
Consider a one-way communication problem $\mathcal{P}$ that admits an exponential quantum advantage achieved by some quantum strategy. Then the number of inputs encoded into magic states in this protocol must grow doubly exponentially with the qudit communication cost of the strategy. In particular, denoting the number of magic states used in the protocol to be $m(n)$, we have 
\begin{align}
m(n) \ge d^{d^{\Omega\left(\mathcal{Q}_{\epsilon}(n)\right)}}.
\end{align}
\end{theorem}

\begin{proof}
The proof proceeds by demonstrating that if $m(n)$ is bounded by any exponential function of the form $d^{O\left(\mathcal{Q}_{\epsilon}^{\,\delta}(n)\right)}$ for a constant $\delta > 0$, the quantum protocol admits a classical simulation that exactly reproduces all achievable input--output correlations $p(b|x,y)$ with only polynomially larger communication cost.

To establish the upper bound, we observe that the $(d^{n}-m(n))$ inputs encoded into stabilizer states can be classically specified using at most polynomially many cdits relative to the quantum communication cost $\mathcal{Q}_{\epsilon}(n)$. By the qudit generalization of the Gottesman-Knill theorem, the efficient classical description of these states allows the stabilizer component of the protocol to be transmitted with a polynomial overhead of $O(\mathcal{Q}_{\epsilon}^2(n))$ cdits.

For the remaining $m(n)$ inputs encoded into magic states, the protocol and the finite set of possible resource states are known \textit{a priori} to both parties. Alice therefore only needs to transmit a classical index uniquely identifying her chosen state. The communication cost to send this label scales logarithmically with the number of magic states, requiring exactly $\lceil \log_d m(n) \rceil$ cdits. 

In either of the case, upon receiving the classical information, Bob gets to know exactly what state Alice intended to send. Bob then utilizes his unrestricted local computational resources to simulate the statistics $p(b|x,y) = \text{Tr}(\Pi_y \rho_x)$ associated with the corresponding measurement. 

Combining these observations, the total classical communication required to simulate the quantum protocol scales as $\max\{O(\mathcal{Q}_{\epsilon}^2(n)), O(\log_d m(n))\}$. If $m(n) \le d^{O\left(\mathcal{Q}_{\epsilon}^{\,\delta}(n)\right)}$, the cost to transmit the index is at most $O(\mathcal{Q}_{\epsilon}^{\,\delta}(n))$ cdits, which remains polynomial in $\mathcal{Q}_{\epsilon}(n)$. 

Therefore, to force the classical communication cost to scale exponentially as $d^{\Omega(\mathcal{Q}_{\epsilon}(n))}$ and preserve the exponential quantum advantage, the number of magic-encoded inputs must scale doubly exponentially, yielding $m(n) \ge d^{d^{\Omega\left(\mathcal{Q}_{\epsilon}(n)\right)}}$.
\end{proof}
Theorem~\ref{theorem3} has a particularly strong implication for communication tasks satisfying 
\begin{align}
\mathcal{C}(n)=O(n), \qquad \mathcal{Q}(n)=O(\log_d n).  
\end{align}
In these cases, although an exponential quantum advantage is achievable, such an advantage can arise only if essentially all \(d^n\) possible inputs are encoded into magic states. Equivalently, any exponential separation necessarily requires an exponential amount of magic-state encoding. In the qubit case (\(d=2\)), the Hidden Matching problem provides a canonical example exhibiting this behaviour. 

\section{Discussions}\label{sec4} Building on the seminal Gottesman–Knill theorem, the importance of non-stabilizer resources (magic) has been extensively studied in the context of quantum computational advantage, classical simulation bounds \cite{Gottesman99,Howard2014,Bravyi2016,Howard2017,Bravyi2018,Bu2019,Gu2025}, quantum homomorphic encryption \cite{Broadbent2015}, quantum machine learning \cite{Hinsche2023}, and quantum channel capacities \cite{Bu2025}. Adding to this line of work, our Theorem \ref{theo1} identifies magic resources as the necessary ingredient for achieving quantum advantage in one-way communication complexity tasks. This result also extends the Frenkel–Weiner no-go theorem \cite{Frenkel2015} to the communication complexity setting. On the other hand, Lemma \ref{theo2} \& \ref{theo2E} and Theorem \ref{theo2s} show that even a minimal magic resource in the encoding suffices to outperform classical protocols in certain communication complexity tasks. This, within prepare-and-measure scenario, opens an experimental pathway for certifying the presence of magic resources. Going beyond these observations, our Theorem \ref{theorem3} shows that achieving an exponential quantum advantage in a communication task necessarily requires an exponential amount of magic in the encoding stage. This establishes a fundamental lower bound on the non-stabilizer resources required for any quantum protocol that exhibits an exponential separation from classical communication.

Our work also opens several interesting directions for future research. Firstly, exploring the status of Theorem \ref{theo1} in communication complexity tasks involving multiple distant parties remains an open problem. Additionally, Theorem \ref{theo2s} focuses on $2$-level RACs only. Extending this analysis to higher-level RACs and, more broadly, to other classes of communication complexity tasks could yield further valuable insights. Finally, while Theorem~\ref{theorem3} characterizes the minimum fraction of inputs that must be encoded into magic states to realize an exponential quantum advantage, identifying concrete communication tasks that achieve such separations with minimal magic consumption per input remains an important direction for future investigation.

\bigskip
\noindent {\bf Note added:} After completing our study, we became aware of an independent line of work \cite{Zamora2025} which studies semi-device-independent techniques within Prepare-and-Measure framework to certify non-stabilizer quantum states. 
\bigskip
\begin{acknowledgments}
\noindent{\bf Acknowledgement:} SRC acknowledges support from University Grants Commission, India (reference no. 211610113404). SGN acknowledges support from the CSIR project $09/0575(15951)/2022$-EMR-I. RKP acknowledges support from the EU (CHIST-ERA MoDIC) and from the National Research, Development and Innovation Office NKFIH (No. $2023-1.2.1-$ERA\_NET$-2023-00009$). SBG, MB, and AGM acknowledge the financial support through the National Quantum Mission (NQM) of the Department of Science and Technology, Government of India. PG acknowledges financial support from the project entitled “Technology Vertical - Quantum Communication” under the National Quantum Mission (NQM) of the Department of Science and Technology (DST) (Sanction Order No. DST/QTC/NQM/QComm/2024/2 (G)).
\end{acknowledgments}


%


\end{document}